\documentclass[aps,superscriptaddress,twocolumn,twoside,floatfix,pra,a4paper]{revtex4-2} 

\usepackage[T1]{fontenc}
\usepackage{filecontents}

\usepackage{times}
\usepackage{epsfig}
\usepackage{amsfonts}
\usepackage{amsmath}
\usepackage{amssymb,amsthm}
\usepackage{xcolor,colortbl}
\usepackage{multirow}
\usepackage{braket}
\usepackage{latexsym}
\usepackage{amsfonts}
\usepackage{mathrsfs}
\usepackage{natbib}
\usepackage{verbatim}
\usepackage{gensymb}
\usepackage{subcaption}
\usepackage{ragged2e}
\usepackage{oplotsymbl}
\usepackage{array}
\usepackage{caption}

\usepackage{tikz}

\usepackage{soul}
\DeclareCaptionJustification{justified}{\justifying}
\usepackage{blkarray}
 \usepackage{graphicx}
 \usepackage{braket}
\newtheorem{theorem}{Theorem}

\newtheorem{definition}{Definition}
\newtheorem{proposition}{Proposition}
\newtheorem{observation}{Observation}

\usepackage[colorlinks=true,linkcolor=red,citecolor=blue,urlcolor=blue]{hyperref}
\allowdisplaybreaks

\newtheorem{Lemma}{Lemma}
\newtheorem{remark}{Remark}

\usepackage{pgfplots}
\usepackage{tikz-3dplot}

\usepackage{sansmath}
\usetikzlibrary{shadings,intersections}
\usetikzlibrary{calc,patterns,angles,quotes}
 
\tdplotsetmaincoords{60}{115}
\pgfplotsset{compat=newest}
\usetikzlibrary{backgrounds}

\newcommand{\Tr}{\operatorname{Tr}}

\newtheorem{manualtheoreminner}{Theorem}
\newenvironment{manualtheorem}[1]{%
  \IfBlankTF{#1}
    {}
    {\renewcommand\themanualtheoreminner{#1}}%
  \manualtheoreminner
}{\endmanualtheoreminner}

\begin{document}

\title{No-Go Theorem for Generic Simulation of Qubit Channels with Finite Classical Resources}

\author{Sahil Gopalkrishna Naik}
\affiliation{Department of Physics of Complex Systems, S. N. Bose National Center for Basic Sciences, Block JD, Sector III, Salt Lake, Kolkata 700106, India.}

\author{Mani Zartab}
\affiliation{F\'{\i}sica Te\`{o}rica: Informaci\'{o} i Fen\`{o}mens Qu\`{a}ntics, Departament de F\'{\i}sica, Universitat Aut\`{o}noma de Barcelona, 08193 Bellaterra (Barcelona), Spain.}

\author{Nicolas Gisin}
\affiliation{Group of Applied Physics, University of Geneva, 1211 Geneva 4, Switzerland.}
\affiliation{Constructor University, Bremen, Germany.}

\author{Manik Banik}
\affiliation{Department of Physics of Complex Systems, S. N. Bose National Center for Basic Sciences, Block JD, Sector III, Salt Lake, Kolkata 700106, India.}

\begin{abstract}
The mathematical framework of quantum theory, though fundamentally distinct from classical physics, raises the question of whether quantum processes can be efficiently simulated using classical resources. For instance, a sender (Alice) possessing the classical description of a qubit state can simulate the action of a qubit channel through finite classical communication with a receiver (Bob), enabling Bob to reproduce measurement statistics for any observable on the state. In this work, we contend that a more general simulation requires reproducing statistics of joint measurements, potentially involving entangled effects, on Alice's system and an additional system held by Bob—even when Bob’s system state is unknown or entangled with a larger system. Within this broad framework, we prove that no finite amount of classical messaging—regardless of how many rounds are used or how large each message can be—can reproduce a perfect qubit channel, highlighting an inescapable barrier in quantum channel simulation with classical resources. We also establish that entangled effects crucially underlies this no-go result. However, for noisy qubit channels, such as those with depolarizing noise, we demonstrate that general simulation is achievable with finite communication. Notably, the required communication increases as the noise decreases, revealing an intricate relationship between the noise in the channel and the resources necessary for its classical simulation.
\end{abstract}

\maketitle
\section{Introduction}
Classical physics, rooted in intuitive and objective principles, offers deterministic descriptions of the physical phenomena we encounter in daily life (though see \cite{Gisin2019}). In stark contrast, the quantum realm defies classical reasoning, exhibiting phenomena that challenge conventional intuition. Quantum mechanics—formulated within the Hilbert space framework—delivers an extraordinarily precise mathematical account of these phenomena, but it refrains from offering clear physical intuition about their nature \cite{Dirac1930, vonNeumann2018, Peres2002}. Nonetheless, the advent of quantum information theory has highlighted practical advantages of quantum resources over their classical counterparts in tasks such as computation, communication, and cryptography \cite{Deutsch1992, Shor94, Grover1996, Bennett1992, Bennett1993, Buhrman2010, Bennett2014, Ekert1991, Gisin2002}. In this context, simulating quantum processes with classical resources promises a compelling research avenue \cite{Feynman1982,Bremner2010,Rahimi2016}. Such investigations serve a dual purpose: quantifying the computational and communicational power of quantum resources while deepening our understanding of the unique features that distinguish quantum phenomena from classical intuitions.

A hallmark of quantum mechanics, underscored by Bell’s theorem \cite{Bell1964} and corroborated through decades of experiments \cite{Freedman1972, Aspect1981, Aspect1982(1), Aspect1982(2), Zukowski1993, Tittel1998,Weihs1998}, is the emergence of nonlocal correlations among the outcomes of local measurements performed on entangled states. These correlations defy any {\it local realistic} explanation \cite{Bell1966, Mermin1993,Aspect2002,Brunner2014,Gisin2023}. Furthermore, entangled states shared among distant parties cannot be prepared through local quantum operations and classical communication (LOCC) \cite{Horodecki2009}. Despite their inherent nonlocality, the local measurement statistics of entangled states can often be faithfully reproduced through finite classical communication between distant parties holding parts of the composite system \cite{Brassard1999, Steiner2000, Massar2001, Regev2010, Branciard2011, Kar2011, Branciard2012, Banik2012, Roy2014,Brassard2019}. This paradigm extends naturally to quantum channel simulation, where a receiver (Bob) aims to replicate the statistics of arbitrary measurements on a quantum state unknown to him but fully known to a sender (Alice), who aids Bob while minimizing the classical communication required \cite{Cerf2000, Toner2003, Methot2004, Montina2012(1), Renner2023}. In particular, the result by Toner and Bacon demonstrated that the statistics of any projective measurement, also called the von Neumann measurement, on a qubit state can be simulated using just two classical bits of communication \cite{Toner2003}. Subsequent work extended this result to more general settings, including positive operator-valued measures (POVMs) \cite{Kraus1983}, further illustrating the feasibility of classical simulation with finite communication \cite{Renner2023}.

In this work, we argue that classical simulation of quantum channels must extend beyond reproducing local measurement statistics to encompass more general scenarios. Specifically, simulations should replicate the statistics of joint measurements—including those in entangled bases—on Alice’s system and an ancillary system held by Bob, even when Bob’s system is unknown or entangled with an external reference. Focusing on such entangled measurements, we prove that a perfect qubit channel cannot be simulated under this paradigm using any finite amount of classical communication. While the prior studies on quantum channel simulation \cite{Montina2012(1), Toner2003, Cerf2000, Methot2004, Renner2023} consider one-way communication protocols from Alice to Bob, more general simulation strategies allow for interactive protocols involving multiple rounds of classical communication. In such settings, Bob may send information to Alice, who responds based on her input and Bob’s message, and so on—provided the process concludes after a finite number of rounds. Notably, our no-go theorem holds even in this most general multi-round framework, establishing a fundamental limit on classical simulations of quantum channels.

On the other hand, when the joint measurements are restricted to product or separable effects, we show that such simulations are always possible with finite classical communication from Alice to Bob, regardless of the quantum state held by Bob—so long as both systems are finite-dimensional. We also study noisy qubit channels and demonstrate that even under arbitrary joint measurements—including entangled ones—classical simulation becomes feasible for qubit depolarizing channels. Importantly, the amount of classical communication required grows as the noise decreases, revealing a delicate trade-off between the coherence of the quantum channel and the classical resources needed for its simulation.

\section{Classical simulation of quantum channels} 
A quantum channel is a physical device—such as an optical fiber—that transmits quantum states from a sender to a receiver, even when the state is unknown or part of a larger entangled system. Mathematically, a quantum channel is represented by a completely positive trace-preserving (CPTP) map \cite{Kraus1983,Self1}. The quantum teleportation protocol demonstrates that such a channel can be perfectly simulated using classical communication, provided the parties share prior entanglement \cite{Bennett1993}. In contrast, when the quantum state is known to the sender—as in Remote State Preparation (RSP) \cite{Lo2000, Pati2000, Bennett2001}—the simulation may be achieved using classical communication supplemented by shared randomness.

Classical simulation of a quantum channel has been studied in \cite{Montina2012(1), Toner2003, Cerf2000, Methot2004, Renner2023}. Formally, Alice, given the classical description of a quantum state $\ket{\psi} \in \mathbb{C}^d$, samples a message $m$ from a distribution $p(m|x,\psi)$ that depends on a shared random variable $x\sim p(x)$. She sends $m$ to Bob, who then aims to reproduce the statistics of an arbitrary POVM $\mathrm{M} = \{\mathrm{E}^k\}$, using a conditional distribution $p(k|m,x,\mathrm{M})$. The simulation must succeed even when Alice is unaware of Bob’s measurement choice. For a qubit, the state projector $\mathrm{P}_{\hat{\psi}} = \frac{1}{2}(\mathbf{I}_2 + \hat{\psi} \cdot \sigma)$ is fully determined by the Bloch vector $\hat{\psi} \in \mathbb{R}^3$. A simulation is successful if, for all $\psi$ and $\mathrm{M}$, $\sum_m \int dx \, p(k|m,x,\mathrm{M}) \, p(m|x,\psi) \, p(x) = \Tr[\Lambda(\mathrm{P}_\psi)\mathrm{E}^k]$, where $\Lambda: \mathcal{D}(\mathbb{C}^d) \to \mathcal{D}(\mathbb{C}^d)$ denotes the quantum channel and $\mathcal{D}(\cdot)$ is the set of density matrices. Importantly, the existing channel simulation protocols consider only one-round protocols with communication from Alice to Bob; and communication cost of such a protocol is defined as the minimum communication required for exact simulation \cite{Winter2002, Harsha2010}. In general, however, one may allow multi-round interactive protocols involving back-and-forth classical communication, which we explore in subsequent sections.

\section{Channel simulation, generic setup} 
The simulation of a quantum channel must faithfully reproduce all possible measurement outcome statistics as dictated by the Born rule. This encompasses broader scenario of channel simulation that involves reproducing the statistics of a joint measurement \(\mathrm{M}_{AB} \equiv \{\mathrm{E}^k_{AB}\}\), applied to system \(A\), known to Alice, and system \(B\), provided to Bob (see Fig. \ref{fig1}). Importantly, the state of system \(B\) may be unknown to both parties and could even form part of a larger entangled system, such as \(BC\). This generalized context introduces significant challenges for classical simulation protocols. As a relevant aside, here we recall the semi-quantum Bell scenario introduced in \cite{Buscemi2012}, which replaces the classical inputs in standard Bell scenario by quantum states that can be nonorthogonal in general (see also \cite{Branciard2013, Banik2013, Chaturvedi2015, Lobo2022}). The quantumness of preparation also find application in verifiability of blind quantum computation \cite{Childs2005,Fitzsimons2017,Ma2022}. While the standard channel simulation scenario, investigated in \cite{Montina2012(1), Toner2003, Cerf2000, Methot2004, Renner2023}, has close analogy with  quantum random access code (QRAC) with the sender and receiver provided with classical inputs (or queries) \cite{Wiesner1983,Ambainis2002}, in our generic scenario receiver’s classical queries are replaced with quantum queries.
\begin{figure}[t!]
\centering
\includegraphics[scale=0.42]{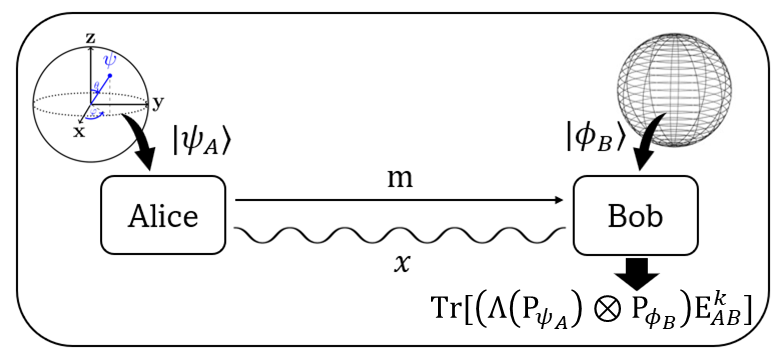}
\caption{The generic simulation of a qubit channel \((\Lambda)\): Alice is given the classical description of a state \(\psi_A \in \mathbb{C}^2_A\), while Bob holds an unknown state \(\phi_B \in \mathbb{C}^2_B\). Their goal is to reproduce the statistics of a joint measurement $\mathrm{M}_{AB}\equiv\{\mathrm{E}^k_{AB}\}$ at Bob's location, which as per Born rule reads as \(\Tr[(\Lambda(\mathrm{P}_{\psi_A}) \otimes \mathrm{P}_{\phi_B}) \mathrm{E}^k_{AB}]\). More generally, the system \(B\) can have arbitrarily large dimension and may also form part of a larger joint system \(BC\).}\vspace{-.2cm}
\label{fig1}
\end{figure}

\section{Results}
Within the generalized channel simulation scenario, we first establish that in one-round simulation the minimum communication cost \(\mathrm{C}^{G}_{\min}\) diverges to infinity, even for the case of qubit channels. 
\begin{theorem}\label{theo1}
The generic simulation of the perfect qubit channel is impossible through one-round protocol by using classical resources alone, even if Alice is permitted to send an arbitrary but finite amount of classical information to Bob.
\end{theorem}
\noindent We defer the proof to the Appendix \ref{apptheo1}. Crucially, any simulation that uses only a finite amount of classical communication—whether one-way or interactive—remains inherently classical. Extending our no-go theorem to general multi-round protocols therefore sharpens our understanding of the classical limits on emulating quantum dynamics. It is worth noting that multi-round protocols involving bidirectional communication are shown to advantageous both in classical and quantum realms. For instance, communication from receiver to sender (called the feedback assistance) can enhence zero-error capacity of noisy classical channels \cite{Shannon1956}, where   entanglement purification protocol under two-way protocol can surpass the one-way optimal bound \cite{Bennett1996}. It is also known that characterizing the multi-round local quantum operation and classical communication protocols is a hard problem \cite{Chitambar2014}. Given these subtleties, generalizing Theorem \ref{theo1} beyond single-round communication is nontrivial. Nonetheless, We now show that even the most powerful finite, multi‑round classical protocol cannot perfectly simulate a qubit channel. Specifically, we allow arbitrarily large but finite number of back‑and‑forth rounds between Alice and Bob, each carrying an arbitrarily large but finite amount of classical data, and refer to such schemes as finite back‑and‑forth classical protocols.
\begin{proposition}\label{prop1}
Any simulation protocol involving finite back-and-forth classical communication can always be implemented with finite classical communication from Alice to Bob, only.   
\end{proposition}
\begin{proof}
(Outline) To establish the no-go result for multi-round case, we begin by noting key structural features of such protocols. Any protocol starting with Bob can be reformulated to begin with Alice by inserting a trivial first message from her. Likewise, since the simulation outcome must be produced at Bob’s end, the final round must involve communication from Alice; any final message from Bob is irrelevant. Thus, without loss of generality, any valid protocol consists of an odd number of rounds, beginning and ending with communication from Alice to Bob. With a finite number of rounds and finite communication per round, the protocol can be represented as a finite tree: each path from root to leaf corresponds to a specific sequence of exchanged messages, $(m_1, m_2, \cdots, m_k)$, with branching determined by Alice’s input and the shared randomness.

The central idea is that, given this finite structure, Alice can locally compute—in advance—the exact message sequence that would be followed for any input state and shared variable. She can then compress this sequence into a single classical message and send it to Bob, thereby collapsing the multi-round protocol into an equivalent one-round protocol. In Appendix \ref{appprop1}, we make this construction explicit for a three-round protocol and show that the resulting correlations can be exactly reproduced by a single-round protocol. By iterating this reduction, any finite-round protocol with bounded communication per round can be simulated by a one-round protocol with finite classical communication from Alice to Bob.
\end{proof}
\noindent Unless stated otherwise, all subsequent results hold for finite back-and-forth classical protocols, and invoking Proposition \ref{prop1}, we can restrict to one-round protocols only. As an immediate consequence of Proposition \ref{prop1}, the no-go result of Theorem \ref{theo1} generalizes as follows:
\begin{theorem}\label{theo2}
The generic simulation of the perfect qubit channel is not possible by any finite back-and-forth classical protocol.
\end{theorem}
\noindent In communication complexity, one quantifies the quantum advantage by the extra classical communication—typically supplemented by preshared randomness—needed to reproduce quantum statistics. Prior results \cite{Montina2012(1), Toner2003, Cerf2000, Methot2004, Renner2023} show that, if Bob receives no quantum state, a finite amount of classical communication can successfully simulate a quantum channel. At first glance, this suggests that quantum–classical gaps might be closed with bounded classical resources. Our Theorem \ref{theo2}, however, demonstrates that even in the simplest nontrivial setting—two‐dimensional Hilbert space ($d=2$)—no finite back‐and‐forth classical protocol can emulate a perfect qubit channel.

\noindent A natural question is whether the use of entangled basis measurements is essential to establish the no-go result in Theorem \ref{theo1}. Specifically, if the joint measurement consists solely of product or separable effects, can its statistics be simulated with finite communication? To address this question we start by observing that the statistics of computational basis measurements, \(\mathrm{M}_{comp}\equiv \{\mathrm{P}_{\hat{z}}\otimes\mathrm{P}_{\hat{z}},\mathrm{P}_{\hat{z}}\otimes\mathrm{P}_{\hat{z}^\perp},\mathrm{P}_{\hat{z}^\perp}\otimes\mathrm{P}_{\hat{z}},\mathrm{P}_{\hat{z}^\perp}\otimes\mathrm{P}_{\hat{z}^\perp}\}\), can be simulated using only 1 bit of classical communication from Alice to Bob: Alice measures her state in \(\sigma_z\) basis and communicates measurement outcome to Bob, who also measures his unknown state in \(\sigma_z\) basis. A similar approach works for the twisted measurement \(\mathrm{M}^{(B)}_{twist}\equiv \{\mathrm{P}_{\hat{z}}\otimes\mathrm{P}_{\hat{z}}, \mathrm{P}_{\hat{z}}\otimes\mathrm{P}_{\hat{z}^\perp}, \mathrm{P}_{\hat{z}^\perp}\otimes\mathrm{P}_{\hat{x}}, \mathrm{P}_{\hat{z}^\perp}\otimes\mathrm{P}_{\hat{x}^\perp}\}\):  Alice measures her state in \(\sigma_z\) basis and communicates the result to Bob, who performs either \(\sigma_z\) or \(\sigma_x\) measurement on his qubit, depending on Alice's communication. The simulation becomes slightly more intricate when the twist is on Alice’s side, namely for the measurement \(\mathrm{M}^{(A)}_{twist}\equiv \{\mathrm{P}_{\hat{z}}\otimes\mathrm{P}_{\hat{z}}, \mathrm{P}_{\hat{z}^\perp}\otimes\mathrm{P}_{\hat{z}}, \mathrm{P}_{\hat{x}}\otimes\mathrm{P}_{\hat{z}^\perp}, \mathrm{P}_{\hat{x}^\perp}\otimes\mathrm{P}_{\hat{z}^\perp}\}\). In this case, Alice performs measurements in both the \(\sigma_z\) and \(\sigma_x\) bases on her state and sends the outcomes to Bob using two separate 1-bit classical channels. This is possible as the state is known to her, and thus she can make copies of it. An alternate protocol is possible using 1-bit of communication from Bob to Alice, followed by 1-bit from Alice to Bob: Bob performs \(\sigma_z\) measurement and communicates the outcome to Alice, who accordingly performs either \(\sigma_z\) or \(\sigma_x\) and communicates back her outcome. Notably,both protocols involve 2-bit of communication. In this regard the following observation is noteworthy.   
\begin{observation}\label{obs1}
Outcome statistics of \(\mathrm{M}^{(A)}_{\text{twist}}\) on Alice's known qubit and Bob's unknown qubit cannot be reproduce at Bob's laboratory with 1 bit of communication. 
\end{observation}
\noindent This follows from the well-established fact that, in the random access code (RAC) task, a qubit outperforms a classical bit \cite{Wiesner1983, Ambainis2002} (see Appendix \ref{appobs1} for details). Considering the most general product von Neumann measurements, we establish the following result (proof discussed in Appendix \ref{apptheo3} along with an explicit example in \ref{apptheo3ex}).
\begin{theorem}\label{theo3}
Statistics of any product von Neumann measurement on a qubit, known to Alice, and an unknown qudit held by Bob can always be simulated at Bob's end by finite classical communication from Alice to Bob.  
\end{theorem}
\begin{figure}[t!]
\begin{tikzpicture}[scale=1.1]
\draw[line width=.5mm] (0,0) circle[radius=2];
---------------
\draw[fill=cyan,opacity=0.5] (0,0) -- ({deg(pi/2)}:2) arc[start angle={deg(pi/2)}, end angle={deg(pi/9)}, radius=2] -- cycle;
\draw[line width=.5mm, blue,->] (0,0) -- ({2*cos(deg(pi/9))},{2*sin(deg(pi/9))});
\node at ({2.3*cos(deg(pi/9))},{2.3*sin(deg(pi/9))}) {\large{$|\beta\rangle$}};
---------------
\draw[fill=cyan,opacity=0.5] (0,0) -- ({deg(pi/2)}:2) arc[start angle={deg(pi/2)}, end angle={deg(8*pi/9)}, radius=2] -- cycle;
\draw[line width=.5mm, blue,->] (0,0) -- ({2*cos(deg(8*pi/9))},{2*sin(deg(8*pi/9))});
\node at ({2.4*cos(deg(8*pi/9))},{2.4*sin(deg(8*pi/9))}) {\large{$|\alpha^\perp\rangle$}};
---------------
\draw[line width=.5mm, red,dash pattern= on 3pt off 5pt,postaction={draw,blue,dash pattern= on 3pt off 5pt,dash phase=4pt,thick},->] (0,0) -- (0,2);
\node at (0,2.25) {\large{$|0\rangle$}};
---------------
\draw[fill=red!15] (0,0) -- ({deg(3*pi/2)}:2) arc[start angle={deg(3*pi/2)}, end angle={deg(10*pi/9)}, radius=2] -- cycle;
\draw[line width=.5mm, red,->] (0,0) -- ({2*cos(deg(10*pi/9))},{2*sin(deg(10*pi/9))});
\node at ({2.4*cos(deg(10*pi/9))},{2.4*sin(deg(10*pi/9))}) {\large{$|\beta^\perp\rangle$}};
---------------
\draw[fill=red!15] (0,0) -- ({deg(3*pi/2)}:2) arc[start angle={deg(3*pi/2)}, end angle={deg(17*pi/9)}, radius=2] -- cycle;
\draw[line width=.5mm, red,->] (0,0) -- ({2*cos(deg(17*pi/9))},{2*sin(deg(17*pi/9))});
\node at ({2.35*cos(deg(17*pi/9))},{2.35*sin(deg(17*pi/9))}) {\large{$|\alpha\rangle$}};
---------------
\draw[line width=.5mm, blue,dash pattern= on 3pt off 5pt,postaction={draw,red,dash pattern= on 3pt off 5pt,dash phase=4pt,thick},->] (0,0) -- (0,-2);
\node at (0,-2.25) {\large{$|1\rangle$}};
---------------
\draw[line width=.2mm, dotted,->] (0,0) -- (2,0);
\node at (2.2,0) {\large{$x$}};
---------------
\node at (2.5,2.2) {\large{$|\beta\rangle:=\sqrt{\frac{2}{3}}|0\rangle+\frac{1}{\sqrt{3}}|1\rangle$}};
\node at (2.5,-2.2) {\large{$|\alpha\rangle:=\sqrt{\frac{2}{3}}|1\rangle+\frac{1}{\sqrt{3}}|0\rangle$}};
\end{tikzpicture}
\caption{(Color online) Twisted-butterfly measurement \(\mathrm{M}_{tb}\): In Alice's part the projectors \(\{\mathrm{P}_{\hat{z}},\mathrm{P}_{\hat{z}^\perp},\mathrm{P}_{\hat{\beta}},\mathrm{P}_{\hat{\alpha}^\perp}\}\) are involved, while in Bob parts the projectors \(\{\mathrm{P}_{\hat{z}},\mathrm{P}_{\hat{z}^\perp},\mathrm{P}_{\hat{\alpha}},\mathrm{P}_{\hat{\beta}^\perp}\}\) are used.}
\vspace{-.5cm}
\label{fig2}
\end{figure}
 
\noindent Theorem \ref{theo3}, however, does not fully resolve the question of whether all product POVMs can be simulated with a finite amount of classical communication from Alice to Bob, as there exist measurements involving only rank-1 product effects, but not LOCC implementable. For instance, consider the following POVM:
\begin{align}
\mathrm{M}_{tb}\equiv\left\{\!\begin{aligned}
&\hspace{1.2cm}\Pi_1:=\mathrm{P}_{\hat{z}}\otimes\mathrm{P}_{\hat{z}^\perp}\\
\Pi_{21}&:=\kappa~\mathrm{P}_{\hat{\alpha}^\perp}\otimes\mathrm{P}_{\hat{z}},\Pi_{22}:=\kappa~\mathrm{P}_{\hat{z}^\perp}\otimes\mathrm{P}_{\hat{\alpha}}\\
\Pi_{31}&:=\kappa~\mathrm{P}_{\hat{\beta}}\otimes\mathrm{P}_{\hat{z}},\Pi_{32}:=\kappa~\mathrm{P}_{\hat{z}^\perp}\otimes\mathrm{P}_{\hat{\beta}^\perp}
\end{aligned}\right\},   
\end{align}
where \(\kappa:=3/4\). We call this the {\it twisted-butterfly} POVM \(\mathrm{M}_{tb}\), a name justified by its structure (see Fig.\ref{fig2}).
\begin{Lemma}\label{lemma1}
The POVM \(\mathrm{M}_{tb}\) is not implementable by Alice and Bob under the operational paradigm of LOCC.    
\end{Lemma}
\begin{proof}
The proof simply follows an argument provided in \cite{Duan2009}. The measurement \(\mathrm{M}_{tb}\) perfectly distinguishes the set of orthonormal states \(S_3\equiv\{\ket{\psi_1}:=\ket{01},~\ket{\psi_2}:=(\ket{\phi^-}-\ket{10})/\sqrt{2},~\ket{\psi_3}:=(\ket{\phi^-}+\ket{10})/\sqrt{2}\}\subset\mathbb{C}^2\otimes\mathbb{C}^2\), as
\(\Tr[\Pi_1\mathrm{P}_{\psi_j}]=\delta_{1j}\) and \(\Tr[(\Pi_{i1}+\Pi_{i2})\mathrm{P}_{\psi_j}]=\delta_{ij}\), for \(i\in\{2,3\}~\wedge~j\in\{1,2,3\}\); here \(\ket{\phi^-}:=(\ket{00}-\ket{11})/\sqrt{2}\). On the other hand, \(\ket{\psi_2}~\&~\ket{\psi_3}\) being entangled, the set \(S_3\) is LOCC indistinguishable \cite{Walgate2002}; and hence proves the claim.  
\end{proof}
\noindent Although the measurement \(\mathrm{M}_{tb}\) is not LOCC implementable, quite interestingly, it turns out that the statistics of this measurement on a qubit state known to Alice and an unknown qubit state provided to Bob can be simulated at Bob's end with finite classical communication from Alice (see Appendix \ref{appbutterfly}). We in fact establish a quite generic result (proof provided in Appendix \ref{apptheo4}).  
\begin{theorem}\label{theo4}
Statistics of any separable measurement on a quantum state known to Alice and an unknown state of another quantum system provided to Bob, can always be simulated at Bob's end by finite classical communication from Alice to Bob. 
\end{theorem}
\noindent It is worth recalling the Gottesman–Knill theorem \cite{Gottesman1998} here, which shows that stabilizer quantum circuits are classically simulable, with quantum advantage arising only through the inclusion of non-stabilizer (`magic’) resources such as T-states \cite{Bravyi205}. In parallel, our Theorem~\ref{theo4} demonstrates that classical simulation of a qubit channel is possible when measurements are restricted to separable ones, thereby requiring entangled basis measurements for exhibiting non-classical behavior. Notably, in our setting, it is the use of a specific entangled basis measurement that is stabilizer-preserving yet leads to non-classical behavior.

Here, we note that Theorem \ref{theo4} admits a natural generalization to multipartite settings. In this scenario, multiple distant senders—Alice-$1$, Alice-$2,\cdots$, Alice-$n$—each receive a classical description of a local quantum state $\psi_{A_i} \in \mathbb{C}^{d_i}_{A_i}$, known only to the $i^\text{th}$ sender. The receiver, Bob, holds an unknown state $\phi_B \in \mathbb{C}^{d}_B$ and aims to reproduce the statistics of a $K$-outcome measurement $\mathrm{M}_{A_1\cdots A_nB} \equiv \{\Pi^b_{A\_1\cdots A_nB}~|~b=1,\cdots, k\}$ on the joint state $\otimes_{i=1}^n \psi_{A_i} \otimes \phi_B$. As shown in Theorems \ref{theo4SA} and \ref{theo4SB} (see Appendix \ref{appmultisep}), the statistics can be simulated using only finite classical communication among the parties whenever each of the effects $\Pi^b_{A_1\cdots A_nB}$'s are fully separable \cite{Horodecki2009}.

\subsection*{Simulating noisy qubit channels} Thus far, we have focused on the simulation of perfect qubit channel. A natural extension is to ask whether the no-go result of Theorem \ref{theo1} applies to imperfect qubit channels. To address this, we consider the qubit depolarizing channel \(\mathrm{D}_\eta: \mathcal{D}(\mathbb{C}^2) \to \mathcal{D}(\mathbb{C}^2)\), defined as \(\mathrm{D}_\eta(\rho) := \eta\rho + (1-\eta)\frac{\mathbf{I}_2}{2}\), where \(\eta \in [0,1]\). We now analyze the classical simulability of this particular class of channels.
\begin{theorem}\label{theo5}
For all \(\eta \in [0,1)\), the qubit depolarizing channel \(\mathrm{D}_\eta\) can be simulated with a finite amount of classical communication from Alice to Bob. The required communication increases as \(\eta \to 1\).    
\end{theorem}
\begin{proof}
Given a known state \(\psi = \frac{1}{2}(\mathbf{I}_2 + \hat{\psi} \cdot \vec{\sigma})\), if Alice can ensure that the state \(\mathrm{D}_\eta(\psi) = \frac{1}{2}(\mathbf{I}_2 + \eta~\hat{\psi} \cdot \vec{\sigma})\) is reproduced at Bob's laboratory, then any generic measurement statistics can also be reproduced by Bob. Let Alice be allowed to communicate \(m\) classical bits to Bob. To reproduced the state \(\mathrm{D}_\eta(\psi)\) at Bob's end their protocol proceeds as follows:- (i) Shared Randomness: Alice and Bob share a classical random variable \(\mathcal{X} \in \mathrm{U}(\mathbb{C}^2)\), which is drawn Haar-randomly from the set of unitary operators on \(\mathbb{C}^2\). (ii) Predefined States: Before the protocol begins, Alice and Bob agree on a set of \(2^m\) equally spaced Bloch vectors \(\{\hat{\omega}_i\}_{i=1}^{2^m}\) with the corresponding qubit states \(\{\omega_i\}_{i=1}^{2^m}\). (iii) Overlap Computation and Communication: Given the input state \(\psi\), Alice computes the overlaps \(\Tr[\mathcal{X} \mathrm{P}_{\hat{\omega}_i} \mathcal{X}^\dagger \mathrm{P}_{\hat{\psi}}]\) for all \(i\) and identifies the index \(i^\star\) that maximizes this overlap. She communicates the index \(i^\star\) to Bob using \(m\)-bit classical communication. (iv) State Preparation at Bob's End: Upon receiving \(i^\star\) and having access to the shared variable \(\mathcal{X}\), Bob prepares the state \(\mathcal{X} \mathrm{P}_{\hat{\omega}_{i^\star}} \mathcal{X}^\dagger\). As shown in Appendix, on average the state \(\mathrm{D}_\eta(\psi)\) is prepared at Bob's laboratory. The parameter \(\eta\) approaches unity as the number of bits \(m\) increases, thus allowing increasingly accurate simulation of the depolarizing channel. 
\end{proof}
\noindent Let \(\eta(m)\) denote the value of the parameter \(\eta\) achieved following the above $m$-bit protocol . In general, deriving an exact expression for \(\eta(m)\) for arbitrary \(m\) is challenging, as it depends on the specific choices of Bloch vectors \(\{\hat{\omega}_i\}_{i=1}^{2^m}\) (see Appendix \ref{apptheo5} for more details). However, for small \(m\)'s we can have some natural choices of Bloch vectors -- \((m\)=\(1)\): 2 diametrically opposite vectors, yielding \(\eta(1) = 1/2\), \((m\)=\(2)\): 4 vectors forming a regular tetrahedron, yielding \(\eta(2) = (3+\sqrt{3})/6 \approx 0.789\), and \((m\)=\(3)\): 8 vectors forming the vertices of a cube, yielding \(\eta(3) = (3+\sqrt{6})/6 \approx 0.908\). 

\section{Discussions} 
No-go theorems play a central role in advancing our understanding of quantum foundations by delineating the boundaries between classical and quantum models and identifying the origins of quantum advantage in information-processing tasks. For instance, the Bell–Kochen–Specker theorem rules out broad classes of classical hidden-variable explanations for quantum statistics \cite{Bell1964, Bell1966, Kochen1967}. Similarly, the Gottesman–Knill theorem shows that stabilizer circuits can be efficiently simulated on a classical probabilistic computer \cite{Gottesman1998}, thereby isolating non-stabilizer elements—so-called ``magic"—as essential for quantum computational speedups \cite{Bravyi205, Howard2014, Campbell2017}. In this spirit, Theorem \ref{theo2} establishes a fundamental no-go result: contrary to prior expectations shaped by earlier studies \cite{Montina2012(1), Toner2003, Cerf2000, Methot2004, Renner2023}, we show that classical simulation of a quantum channel is impossible in more general setup. This reveals a deeper, previously overlooked boundary between classical and quantum resources.

Our result also opens up some foundational questions too. In standard simulation scenarios, it has been shown that reproducing the statistics of any POVM at Bob’s end via finite classical communication from Alice is possible {\it if and only if} there exists a $\psi$-epistemic model underlying quantum theory, wherein the wavefunction reflects an agent's knowledge of an underlying reality \cite{Montina2012(1)}. Extending this framework to our generalized simulation task, Theorem \ref{theo1} \& \ref{theo2} suggest impossibility of such $\psi$-epistemic accounts. Instead, it points toward a $\psi$-ontic interpretation of the quantum state, where the wavefunction corresponds to intrinsic physical properties of the system.

In an ontological model, the preparation of quantum states $\psi$ and $\phi$ for Alice and Bob induces a joint probability distribution $p(\lambda|\psi, \phi)$ over an ontic space $\Lambda$, which can be decomposed as $\Lambda = \Lambda_A \times \Lambda_B \times \Lambda_G$. Here, $\Lambda_A$ and $\Lambda_B$ represent the physical properties locally accessible to Alice and Bob, while $\Lambda_G$ encodes global or holistic features relevant for joint measurements. Natural locality constraints imply that $p(\lambda_A|\psi, \phi)$ is independent of $\phi$, and $p(\lambda_B|\psi, \phi)$ is independent of $\psi$. For classical simulation, it is reasonable to assume that Alice’s message $m$ depends only on $\lambda_A$, not on $\lambda_B$ or $\lambda_G$. Within this setting, our result offers support for the ontic status of the qubit wavefunction, without relying on the Preparation Independence assumption invoked by the PBR theorem \cite{Pusey2012}, which has attracted criticism \cite{Schlosshauer2014}. On the other hand, inspired by insights from \cite{Henderson2000}, it would be of interest to explore whether Theorem \ref{theo1} (and \ref{theo2}) continues to hold when Bob’s unknown state is drawn from a restricted set, rather than the full state space. \\

\begin{acknowledgements}
SGN, NG, and MB acknowledge the conference ``{\it Observing a Century of Quantum Mechanics}" held at IISER Kolkata, where initial discussion of this project started. MZ thanks Giulio Gasbarri and Some Sankar Bhattacharya for their constructive comments. SGN extends sincere gratitude to Mir Alimuddin for the valuable comments and corrections provided on an earlier version of the manuscript. SGN acknowledges support from the CSIR project $09/0575(15951)/2022$-EMR-I. MZ acknowledges support from Ministerio de Ciencia e Innovación of the Spanish Government PRE2020-093634 and project PID2022-141283NB-I00 funded by MICIU/AEI/10.13039/501100011033. MZ also acknowledges MCIN with funding from European Union NextGenerationEU (PRTR-C17.I1) and by Generalitat de Catalunya,  the Ministry of Economic Affairs and Digital Transformation of the Spanish Government through the QUANTUM ENIA project: Quantum Spain, by the European Union through the Recovery, Transformation and Resilience Plan - NextGenerationEU within the framework of the “Digital Spain 2026 Agenda”. NG acknowledges support from the NCCR SwissMap. MB acknowledges the financial support through the National Quantum Mission (NQM) of the Department of Science and Technology, Government of India.
\end{acknowledgements}

\onecolumngrid
\section*{Appendix}
\section{Proof of Theorem \ref{theo1}}\label{apptheo1}
\begin{proof}
Consider that the classical description of a qubit state, \(\psi:= \frac{1}{2}(\mathbf{I}_2 + \hat{\psi} \cdot \vec{\sigma})\), is provided to Alice. Meanwhile, Bob is provided with another qubit state \(\phi := \frac{1}{2}(\mathbf{I}_2 + \hat{\phi} \cdot \vec{\sigma})\), which is unknown to both Alice and Bob. While generic simulation requires reproducing the statistics of all possible POVMs \(\mathrm{M}_{AB} \equiv \{\mathrm{E}^k_{AB}\}_k\) performed on the joint state \(\psi_A \otimes \phi_B\) even when the POVM is not known to Alice, let us focus on a specific measurement, \(\mathrm{M}_{\text{singlet}} \equiv \mathrm{M}_{\psi^-} := \{\mathrm{P}_{\psi^-}, \mathbf{I}_4 - \mathrm{P}_{\psi^-}\}\) only; where \(\ket{\psi^-} := \frac{1}{\sqrt{2}}(\ket{01} - \ket{10})\) is the singlet state. According to the Born rule, the probability of obtaining the outcome \(\ket{\psi^-}\) is given by:
\begin{align}
p_{\psi,\phi} := p(\psi^-|\psi, \phi, \mathrm{M}_{\psi^-}) = \frac{1}{4}(1 - \hat{\psi} \cdot \hat{\phi}). \label{eq1}  
\end{align}
The most general classical protocol that Alice and Bob can implement to simulate the statistics in Eq.~(\ref{eq1}) proceeds as follows: Alice generates a classical variable \(m\), sampled according to the conditional distribution \(p(m|x, \psi)\), where \(x\) is a shared random variable sampled as \(p(x)\), and \(\psi\) is the state given to Alice, uniformly sampled from the Bloch sphere; Alice communicates \(m\) to Bob; Bob performs a two-outcome POVM \(\mathrm{M}^{m,x} \equiv \{\mathrm{E}^{m,x}, \mathbf{I}_2 - \mathrm{E}^{m,x}\}\) on the unknown state \(\phi\). Since \(\phi\) is unknown to both parties, their protocols are independent of \(\phi\). Associating \(\mathrm{E}^{m,x}\) with the \(\psi^-\) outcome in Eq.~(\ref{eq1}), perfect simulation demands
\begin{align*}
p_{\psi,\phi}= \sum_m \int dx~p(x)~p(m|x,\psi)~\langle \phi | \mathrm{E}^{m,x} | \phi \rangle= \langle \phi | \mathrm{F}_\psi | \phi \rangle,
\end{align*}
where \(\mathrm{F}_\psi := \sum_m \int dx~p(x)~p(m|x,\psi)~\mathrm{E}^{m,x}\) is the effective POVM element implemented by Bob on \(\phi\), given that Alice is provided with \(\psi\).

Consider now the case where \(\phi = \psi\), leading to \(p_{\psi,\psi} = 0\), which implies \(\mathrm{F}_\psi = \beta_\psi \mathrm{P}_{\hat{\psi}^\perp}\), with \(\beta_\psi \geq 0\) for all \(\psi\); here \(\ket{\psi^\perp}\) denotes the state orthonormal to \(\ket{\psi}\), with Bloch vectors \(\hat{\psi}^\perp=-\hat{\psi}\). Next, consider the case where \(\phi = \psi^\perp\); here, \(p_{\psi,\psi^\perp} = \frac{1}{2}\), implying \(\beta_\psi = \frac{1}{2}\) for all \(\psi\). Thus, whenever Alice is given the state \(\psi\), Bob's effective POVM on \(\phi\) must take the form \(\left\{\mathrm{F}_\psi, \mathbf{I}_2 - \mathrm{F}_\psi\right\}\), where
\begin{align}\label{maineq}
\mathrm{F}_\psi = \frac{1}{2}\mathrm{P}_{\hat{\psi}^\perp}= \sum_m \int dx~p(x)~p(m|x,\psi)~\mathrm{E}^{m,x}.
\end{align}
Assume that $M$ is the cardinality of the message set $\{m\}_{m=0}^{M-1}$. Notably, without loss of any generality we can assume $p(x) > 0$ for all $x$. Let us now define the set
\begin{align}
\mathcal{P}:=\{(m,x)~|~p(m|x,\psi)>0~\text{for some}~\psi\}.
\end{align}
Since the left-hand side of Eq.~(\ref{maineq}) is a rank-1 operator for all $\psi$, it follows that $\mathrm{E}^{m,x}$ must have rank at most 1 and can be expressed as $\mathrm{E}^{m,x} = e^{m,x} \ket{\chi^{m,x}}\bra{\chi^{m,x}}$, where $0 \leq e^{m,x} \leq 1$ for all $(m,x) \in \mathcal{P}$. Furthermore, we can assume that $\mathrm{E}^{m,x} = e^{m,x} \ket{\chi^{m,x}}\bra{\chi^{m,x}}$ holds even for $(m,x) \notin \mathcal{P}$, since such terms do not contribute to the right-hand side of Eq.~(\ref{maineq}) for any $\psi$. Substituting this decomposition into Eq.~(\ref{maineq}), we obtain
\begin{align*}
\frac{1}{2}\mathrm{P}_{\hat{\psi}^\perp}= \sum_m \int dx~p(x)~p(m|x,\psi)~e^{m,x} \ket{\chi^{m,x}}\bra{\chi^{m,x}}.
\end{align*}
Next, we define the set $\Lambda^m_{\psi}:= \{x~|~p(m|x,\psi) e^{m,x} > 0\}$. Thus, we can rewrite Eq.~(\ref{maineq}) as
\begin{align}\label{maineq1}
\frac{1}{2}\mathrm{P}_{\hat{\psi}^\perp}= \sum_m \int_{\Lambda^m_{\psi}} dx~p(x)~p(m|x,\psi)~e^{m,x} \ket{\chi^{m,x}}\bra{\chi^{m,x}},
\end{align}
Since the left-hand side of Eq.~(\ref{maineq1}) is proportional to $\mathrm{P}_{\hat{\psi}^\perp}$, it follows that $\ket{\chi^{m,x}}\bra{\chi^{m,x}} = \mathrm{P}_{\hat{\psi}^\perp},~\forall~x\in\Lambda^m_{\psi}$. Thus, we must also have $\Lambda^m_{\psi} \cap \Lambda^m_{\psi^{\prime}} = \emptyset,~ \forall~m~,~\psi\neq\psi^{\prime}$.
Taking the trace on both sides of Eq.~(\ref{maineq1}), we obtain
\begin{align}\label{maineq2}
\frac{1}{2} &= \sum_m \int_{\Lambda^m_{\psi}} dx~p(x)~p(m|x,\psi)~e^{m,x} \leq \sum_m \int_{\Lambda^m_{\psi}} dx~p(x) \nonumber\\
&\leq M \times \max_{m} \left\{\int_{\Lambda^m_{\psi}} dx~p(x) \right\},~\forall~\psi.
\end{align}
Since there are uncountably many values of $\psi$ and only finitely many values of $m$, by the pigeonhole principle, there must exist some $m_0$ such that $\int_{\Lambda^{m}_{\psi}} dx~p(x)$ attains its maximum for $m=m_0$ for uncountably many values of $\psi$. Consequently, we must have
\begin{align}\label{ineq}
\frac{1}{2M} \leq \int_{\Lambda^{m_0}_{\psi}} dx~p(x),
\end{align}
for uncountably many values of $\psi$, while also satisfying the disjoint-ness condition:
\begin{align}\label{disjoint}
\Lambda^{m_0}_{\psi} \cap \Lambda^{m_0}_{\psi^{\prime}} = \emptyset, \quad \forall~\psi\neq\psi^{\prime}.
\end{align}
However, conditions (\ref{ineq}) and (\ref{disjoint}) are impossible to satisfy for any finite $M$, as the total probability must satisfy $\int dx~p(x) = 1$. Thus, we arrive at a contradiction, proving that it is impossible to simulate a qubit with finite classical communication. This completes the proof. 
\end{proof}

\section{Detailed proof of Proposition \ref{prop1}}\label{appprop1}
We first start by proving any $3$ round protocol with finite communication in every round can be implemented by a $1$ way finite communication protocol 
\subsection{3-way communication protocol}\label{appendix3way}
Let $x$ be the shared variable between Alice and Bob. Alice starts by tossing a coin $q(m_1|\psi x)$ and communicating $m_1$ to Bob. Bob performs a measurement $M_{m_1x}\equiv\{M^{m_2}_{m_1x}\}$. Here writing general measurement operators is necessary, as POVM elements do not uniquely specify the post measurement state, which is essential for knowing statistics for further measurements. We have $\sum_{m_2}M^{m_2\dagger}_{m_1x}M^{m_2}_{m_1x}=\mathbb{I}~\forall~m_1,x$. Next Bob communicated his outcome $m_2$ back to Alice and then Alice tosses the coin $r(m_3|m_1m_2\psi x)$ and communicates $m_3$. Bob then performs a POVM measurement $M_{m_1m_2m_3x}=\{\pi^b_{m_1m_2m_3x}\}$ on the updated state. Here $\sum_b \pi^b_{m_1m_2m_3x}=\mathbb{I}~\forall~m_1,m_2,m_3,x$. Mathematically the correlation $p(b|\psi\phi)$ can be written as follows:
\begin{subequations}
\begin{align}
p(b|\psi\phi)&=\sum_{m_1m_2m_3x}p(bm_1m_2m_3x|\psi\phi)=\sum_{m_1m_2m_3x}p(x|\psi\phi)p(bm_1m_2m_3|\psi\phi x)\nonumber\\
&=\sum_{m_1m_2m_3 x}\mu( x)p(m_1|\psi\phi x)p(bm_2m_3|m_1\psi\phi x)\nonumber\\
&=\sum_{m_1m_2m_3 x}\mu( x)q(m_1|\psi x)p(m_2|m_1\psi\phi x)p(bm_3|m_1m_2\psi\phi x)\nonumber\\
&=\sum_{m_1m_2m_3 x}\mu( x)q(m_1|\psi x)~\Tr[M^{m_2}_{m_1 x}\phi M^{m_2\dagger}_{m_1 x}]~p(bm_3|m_1m_2\psi\phi x)\nonumber\\
&=\sum_{m_1m_2m_3 x}\mu( x)q(m_1|\psi x)~\Tr[M^{m_2}_{m_1 x}\phi M^{m_2\dagger}_{m_1 x}]~p(m_3|m_1m_2\psi\phi x)p(b|m_1m_2m_3\psi\phi x)\nonumber\\
&=\sum_{m_1m_2m_3 x}\mu( x)q(m_1|\psi x)~\Tr[M^{m_2}_{m_1 x}\phi M^{m_2\dagger}_{m_1 x}]~p(m_3|m_1m_2\psi\phi x)p(b|m_1m_2m_3\psi\phi x)\nonumber\\
&=\sum_{m_1m_2m_3 x}\mu( x)q(m_1|\psi x)~\Tr[M^{m_2}_{m_1 x}\phi M^{m_2\dagger}_{m_1 x}]~r(m_3|m_1m_2\psi x)p(b|m_1m_2m_3\psi\phi x)\nonumber\\
&=\sum_{m_1m_2m_3 x}\mu( x)q(m_1|\psi x)~\Tr[M^{m_2}_{m_1 x}\phi M^{m_2\dagger}_{m_1 x}]~r(m_3|m_1m_2\psi x)\times\nonumber\\
&\hspace{5cm}\left\{\Tr\left[\pi^b_{m_1m_2m_3 x}\left(\frac{M^{m_2}_{m_1 x}\phi M^{m_2\dagger}_{m_1 x}}{\Tr[M^{m_2}_{m_1 x}\phi M^{m_2\dagger}_{m_1 x}]}\right)\right]\right\}\nonumber\\
&=\sum_{m_1m_2m_3 x}\mu( x)q(m_1|\psi x)r(m_3|m_1m_2\psi x)\left\{\Tr\left[\pi^b_{m_1m_2m_3 x}\left(M^{m_2}_{m_1 x}\phi M^{m_2\dagger}_{m_1 x}\right)\right]\right\}\label{3wayeq}
\end{align}    
\end{subequations}
In the next subsection we show that the above correlation can be simulated by finite $1$ way protocol. 
\subsection{1-way communication protocol to simulate the 3-way protocol in \ref{appendix3way}} 
Here Alice tosses the following coins $q(m_1|\psi x),r(\eta_1|\psi x m_1,m_2=1),r(\eta_2|\psi x m_1,m_2=2),\cdots,r(\eta_k|\psi x m_1,m_2=k)$. Here we have assumed $m_2$ takes values in the set $\{1,\cdots,k\}$. Alice then communicates the outcomes $m_1,\eta_1,\cdots,\eta_k$. Bob follows the same protocol for first measurement and gets an outcome $m_2$ and then he generates $m_3=\eta_{m_2}$ and then follows the same protocol as in the $3$ way case. Here we aim to prove that this particular protocol yields the same $p(b|\psi\phi)$ as in Eq.(\ref{3wayeq}). 
\begin{subequations}
\begin{align}
p(b|\psi\phi)&=\sum_{m_1 x}p(bm_1 x|\psi\phi)=\sum_{m_1 x}\mu( x)q(m_1|\psi x)p(b|m_1\psi\phi x)\label{2c} 
\end{align}
\end{subequations}
Now we try to write $p(b|m_1\psi\phi x)$
\begin{subequations}
\begin{align}
p(b|m_1\psi\phi x)&=\sum_{\eta_1\cdots\eta_k}p(b\eta_1\cdots\eta_k|m_1\psi\phi x)=\sum_{\eta_1\cdots\eta_k}p(\eta_1|m_1\psi\phi x)p(b\eta_2\cdots\eta_k|m_1\psi\phi x\eta_1)\nonumber\\
&=\sum_{\eta_1\cdots\eta_k}r(\eta_1|\psi x m_1,m_2=1)p(\eta_2|m_1\psi\phi x\eta_1)p(b\eta_3\cdots\eta_k|m_1\psi\phi x\eta_1\eta_2)\nonumber\\
&=\sum_{\eta_1\cdots\eta_k}r(\eta_1|\psi x m_1,m_2=1)r(\eta_2|\psi x m_1,m_2=2)p(b\eta_3\cdots\eta_k|m_1\psi\phi x\eta_1\eta_2)\nonumber
\end{align}
Repeating the same procedure we get
\begin{align}
p(b|m_1\psi\phi x)&=\sum_{\eta_1\cdots\eta_k}\left\{\prod_{i=1}^kr(\eta_i|\psi x m_1,m_2=i)\right\}p(b|m_1\phi x\eta_1\cdots\eta_k)\label{3e}
\end{align}
\end{subequations}
Now we compute $p(b|m_1\phi x\eta_1\cdots\eta_k)$
\begin{subequations}
\begin{align}
p(b|m_1\phi x\eta_1\cdots\eta_k)&=\sum_{m_2}p(m_2b|m_1\phi x\eta_1\cdots\eta_k)\nonumber\\
&=\sum_{m_2}p(m_2|m_1\phi x\eta_1\cdots\eta_k)p(b|m_1m_2\phi x\eta_1\cdots\eta_k)\nonumber\\
&=\sum_{m_2}p(m_2|m_1\phi x)p(b|m_1m_2\phi x\eta_1\cdots\eta_k)\nonumber\\
&=\sum_{m_2}~\Tr[M^{m_2}_{m_1 x}\phi M^{m_2\dagger}_{m_1 x}]~p(b|m_1m_2\phi x\eta_1\cdots\eta_k)\nonumber\\
&=\sum_{m_2}~\Tr[M^{m_2}_{m_1 x}\phi M^{m_2\dagger}_{m_1 x}]~p(b|m_1m_2\phi x\eta_{m_2})\nonumber\\
&=\sum_{m_2}~\Tr[M^{m_2}_{m_1 x}\phi M^{m_2\dagger}_{m_1 x}]~\left\{\Tr\left[\pi^b_{m_1m_2\eta_{m_2} x}\left(\frac{M^{m_2}_{m_1 x}\phi M^{m_2\dagger}_{m_1 x}}{\Tr[M^{m_2}_{m_1 x}\phi M^{m_2\dagger}_{m_1 x}]}\right)\right]\right\}\nonumber\\
&=\sum_{m_2}\Tr\left[\pi^b_{m_1m_2\eta_{m_2} x}\left(M^{m_2}_{m_1 x}\phi M^{m_2\dagger}_{m_1 x}\right)\right]\label{4g}
\end{align}
\end{subequations}
Replacing $m_2$ by another dummy variable $m'_2$ in Eq.(\ref{4g}) and substituting $p(b|m_1\phi x\eta_1\cdots\eta_k)$ in Eq.(\ref{3e}) we get
\begin{subequations}
\begin{align}
p(b|m_1\psi\phi x)&=\sum_{\eta_1\cdots\eta_k}\left\{\prod_{i=1}^kr(\eta_i|\psi x m_1,m_2=i)\right\}\left(\sum_{m'_2}\Tr\left[\pi^b_{m_1m'_2\eta_{m'_2} x}\left(M^{m'_2}_{m_1 x}\phi M^{m'_2\dagger}_{m_1 x}\right)\right]\right)\nonumber\\
&=\sum_{m'_2}\sum_{\eta_1\cdots\eta_k}\left\{\prod_{i=1}^kr(\eta_i|\psi x m_1,m_2=i)\Tr\left[\pi^b_{m_1m'_2\eta_{m'_2} x}\left(M^{m'_2}_{m_1 x}\phi M^{m'_2\dagger}_{m_1 x}\right)\right]\right\}\nonumber\\
&=\sum_{m'_2}\sum_{\eta_1\cdots\eta_k}\left\{\left(\prod_{i\neq m'_2}r(\eta_i|\psi x m_1,m_2=i)\right)~r(\eta_{m'_2}|\psi x m_1,m_2=m'_2)\times\right.\nonumber\\
&\left.\hspace{4cm}\Tr\left[\pi^b_{m_1m'_2\eta_{m'_2} x}\left(M^{m'_2}_{m_1 x}\phi M^{m'_2\dagger}_{m_1 x}\right)\right]\right\}\nonumber\\
&=\sum_{m'_2}\sum_{\eta_{m'_2}}\left\{\left(\prod_{i\neq m'_2}\sum_{\eta_i}r(\eta_i|\psi x m_1,m_2=i)\right)~r(\eta_{m'_2}|\psi x m_1,m_2=m'_2)\times\right.\nonumber\\
&\left.\hspace{4cm}\Tr\left[\pi^b_{m_1m'_2\eta_{m'_2} x}\left(M^{m'_2}_{m_1 x}\phi M^{m'_2\dagger}_{m_1 x}\right)\right]\right\}.\nonumber
\end{align}
Since $\sum_{\eta_i}r(\eta_i|\psi x m_1,m_2=i)=1$  and $\prod_{i\neq m'_2}1=1$ we have
\begin{align}
p(b|m_1\psi\phi x)&=\sum_{m'_2\eta_{m'_2}}r(\eta_{m'_2}|\psi x m_1,m_2=m'_2)~\Tr\left[\pi^b_{m_1m'_2\eta_{m'_2} x}\left(M^{m'_2}_{m_1 x}\phi M^{m'_2\dagger}_{m_1 x}\right)\right]\nonumber\\
p(b|m_1\psi\phi x)&=\sum_{m_2m_3}r(m_3|\psi x m_1m_2)~\Tr\left[\pi^b_{m_1m_2m_3 x}\left(M^{m_2}_{m_1 x}\phi M^{m_2\dagger}_{m_1 x}\right)\right]\label{5f}
\end{align}
\end{subequations}
Where we replaced the dummy variables $m'_2$ and $\eta_{m'_2}$ by the variables $m_2$ and $m_3$ respectively. Now substituting $p(b|m_1\psi\phi x)$ from Eq.(\ref{5f}) in Eq.(\ref{2c}) we get,
\begin{align}
p(b|\psi\phi)
&=\sum_{m_1m_2m_3 x}\mu( x)q(m_1|\psi x)~r(m_3|\psi x m_1m_2)~\Tr\left[\pi^b_{m_1m_2m_3 x}\left(M^{m_2}_{m_1 x}\phi M^{m_2\dagger}_{m_1 x}\right)\right]    
\end{align}
Which exactly matches the expression for the $3$ way case. This proves that any 3 way bounded communication protocol can be simulated by a $1$ way bounded communication protocol. This also proves that we can also convert any bounded back and forth communication to only $1$ way communication.

\section{Proof of Observation \ref{obs1}}\label{appobs1}
\begin{proof}
Arguably, 1-bit of communication from Bob to Alice is not useful. We thus prove that 1-bit of communication from Alice to Bob is not sufficient to simulate the statistics. 

\noindent We start by recalling the $2\mapsto1$ RAC task \cite{Wiesner1983,Ambainis2002}, where Alice is provided with a random bit string \(x_0x_1\in\{0,1\}^2\) and Bob is randomly given \(y\in\{0,1\}\). Bob's aim to produce a 1-bit outcome \(b=x_y\) with the help of 1-bit respectively 1-qubit communication from Alice. Qubit strategies yield the optimal success \(P_Q=1/2(1+1/\sqrt{2})\) which is strictly higher than the optimal c-bit success \(P_C=1/2(1+1/2)\).

Contrary to the claim of the Observation $1$, let us assume that the statistics of the measurement \(\mathrm{M}^{(A)}_{twist}\equiv \{\mathrm{P}_{\hat{z}}\otimes\mathrm{P}_{\hat{z}}, \mathrm{P}_{\hat{z}^\perp}\otimes\mathrm{P}_{\hat{z}}, \mathrm{P}_{\hat{x}}\otimes\mathrm{P}_{\hat{z}^\perp}, \mathrm{P}_{\hat{x}^\perp}\otimes\mathrm{P}_{\hat{z}^\perp}\}\) on a state known to Alice and an unknown state of Bob system can be simulated at Bob's end with just 1-bit of classical communication from Alice to Bob. Let us denote this protocol as \(1\)-\(\mathrm{CBS}\). As we will argue now, this protocol can be utilized to perform the $2\mapsto1$ RAC task. Given the bit string Alice will implement the \(1\)-\(\mathrm{CBS}\) protocol on the preparation 
\begin{align}
\psi^{x_0x_1}_A=\frac{1}{2}\left[\mathbf{I}_2+\frac{1}{\sqrt{2}}\{(-1)^{x_0}\sigma_3+(-1)^{x_1}\sigma_1\}\right],
\end{align}
whereas Bob, given the question \(y\), will prepare the state 
\begin{align}
\phi^{y}_B=\frac{1}{2}\left[\mathbf{I}_2+(-1)^y\sigma_3\right]. 
\end{align}
As per the assumption, \(1\)-\(\mathrm{CBS}\) protocol reproduce the statistics of \(\mathrm{M}^{(A)}_{twist}\) on \(\psi^{x_0x_1}_A\otimes\phi^{y}_B\) at Bob's laboratory. Bob can post-process this outcome statistics and can accordingly devise a strategy to answer his guess \(b\). In particular, for the outcomes \(\mathrm{P}_{\hat{z}}\otimes\mathrm{P}_{\hat{z}}\) and \(\mathrm{P}_{\hat{x}}\otimes\mathrm{P}_{\hat{z}^\perp}\) Bob guesses \(b=0\), else he guesses \(b=1\). Denoting \(\Pi_0=\mathrm{P}_{\hat{z}}\otimes\mathrm{P}_{\hat{z}}+\mathrm{P}_{\hat{x}}\otimes\mathrm{P}_{\hat{z}^\perp}\) and \(\Pi_1=\mathrm{P}_{\hat{z}^\perp}\otimes\mathrm{P}_{\hat{z}}+\mathrm{P}_{\hat{x}^\perp}\otimes\mathrm{P}_{\hat{z}^\perp}\), we have 
\begin{align}
\text{Pr}(b=x_y|x_0x_1,y)&=\Tr\left[\left(\psi^{x_0x_1}_A\otimes\phi^{y}_B\right)\Pi_{b=x_y}\right]=\frac{1}{2}\left(1+\frac{1}{\sqrt{2}}\right),~\forall~x_ox_1~\wedge~y.
\end{align}
Therefore, with \(1\)-\(\mathrm{CBS}\) protocol one can have a success \(P_{1-\mathrm{CBS}}=1/2(1+1/\sqrt{2})\) in \(2\mapsto1\) RAC task -- a contradiction. In other words, this proves that with 1-bit communication the statistics of \(\mathrm{M}^{(A)}_{twist}\) cannot be reproduced at Bob's laboratory.   
\end{proof}

\section{Proof of Theorem \ref{theo3}}\label{apptheo3}
\begin{proof}
It is known that any product von Neumann measurement in \(\mathbb{C}^2\otimes\mathbb{C}^d\) are implementable under LOCC \cite{Bennett1999}. Proof of our theorem follows a similar reasoning as of there. A generic orthonormal product Basis (OPB) of \(\mathbb{C}^2\otimes\mathbb{C}^d\) takes the form $\mathbf{B}=\cup_i\mathbf{B}_i$, with 
\begin{align}
\mathbf{B}_i:=\{\ket{\alpha_i}\otimes\ket{\beta_{ij}}, \ket{\alpha^\perp_i}\otimes\ket{\Tilde{\beta}_{ij}}\},
\end{align}
where, \(\langle\beta_{ij}|\beta_{i'j'}\rangle=\langle\Tilde{\beta}_{ij}|\Tilde{\beta}_{i'j'}\rangle=\delta_{ii'}\delta_{jj'}\) and for \(i\neq i'\), \(\langle\beta_{ij}|\Tilde{\beta}_{i'j'}\rangle=\delta_{jj'}\). Notably, the subspaces \(\mathcal{S}_i=\text{Span}\{\ket{\beta}_{ij}, \ket{\Tilde{\beta}_{ij}}\}_j\) at Bob's part are mutually orthogonal. To simulate the statistic of von Neumann measurement on the basis \(\mathbf{B}\) they apply the following protocol: (i) Bob performs a measurement distinguishing the subspaces \(\mathcal{S}_i\)'s, while Alice performs measurements \(\mathrm{M}_i\equiv\{\mathrm{P}_{\alpha_i}, \mathrm{P}_{\alpha^\perp_i}\}\) on different copies of her known state, and through different classical channels she communicates \(0_i(1_i)\) whenever the projector \(\mathrm{P}_{\alpha_i}(\mathrm{P}^\perp_{\alpha_i})\) clicks; (ii) Bob considers the communication from $i^{th}$ channel if his projector corresponding to \(\mathcal{S}_i\) subspace clicks, and then he performs a measurement that distinguishes the states \(\{\ket{\beta}_{ij}\}\) if Alice's communication is $0_i$, otherwise he performs a measurement that distinguishes the states \(\{\ket{\Tilde{\beta}}_{ij}\}\).  This completes the proof. 
\end{proof}

\subsection{An explicit Example}\label{apptheo3ex}
\noindent For a better appreciation of Theorem $3$, here we provide an explicit example. Consider the OPB \(\mathbf{B}=\mathbf{B}_1\cup\mathbf{B}_2\cup\mathbf{B}_3\) of \(\mathbb{C}^2\otimes\mathbb{C}^6\) system, where 
\begin{align}
\left\{\begin{aligned}
\mathbf{B}_1&=\left\{\ket{0}\ket{0},\ket{0}\ket{1},\ket{1}\ket{x^{01}_+},\ket{1}\ket{x^{01}_-}\right\},\\
\mathbf{B}_2&=\left\{\ket{x^{01}_+}\ket{2},\ket{x^{01}_+}\ket{3},\ket{x^{01}_-}\ket{x^{23}_+},\ket{x^{01}_-}\ket{x^{23}_-}\right\},\\
\mathbf{B}_2&=\left\{\ket{y^{01}_+}\ket{4},\ket{y^{01}_+}\ket{5},\ket{y^{01}_-}\ket{x^{45}_+},\ket{y^{01}_-}\ket{x^{45}_-}\right\}
\end{aligned}\right\},
\end{align}   
with, \(\ket{x^{lm}_{\pm}}:=\frac{1}{\sqrt{2}}(\ket{l}\pm\ket{m})\) and \(\ket{y^{lm}_{\pm}}:=\frac{1}{\sqrt{2}}(\ket{l}\pm\iota\ket{m})\). To simulate statistics of the measurement on this basis, Bob, on his unknown state, first performs a measurement $\mathrm{M}^B_{1^{st}}$ consisting of three rank-2 projective effects, i.e.
\begin{align}
\mathrm{M}^B_{1^{st}}\equiv\left\{\!\begin{aligned}
\mathbb{P}_1&:=\ket{0}\bra{0}+\ket{1}\bra{1},~~\mathbb{P}_2:=\ket{2}\bra{2}+\ket{3}\bra{3},~~
\mathbb{P}_3:=\ket{4}\bra{4}+\ket{5}\bra{5}\\
\end{aligned}\right\}.
\end{align}
On the other hand, Alice performs $\sigma_z$, $\sigma_x$, and $\sigma_y$ measurements on three copies of her known state, and communicates the outcomes through three 1-bit classical channels, respectively $1^{st}$, $2^{nd}$, and $3^{rd}$, to Bob. Depending on which projector clicks in his first measurement, Bob chooses the corresponding communication line from Alice, and depending on the communication received from Alice, he performs the measurements as shown is Table \ref{tab1}.
\begin{table}[h!]
\begin{center}
\begin{tabular}{|c|c|c|c|c}
\cline{1-4}
Outcome of \(\mathrm{M}^B_{1st}\)& Channel& Commun.& Bob's measurement &  \\ \cline{1-4}
\multirow{2}{*}{\(\mathbb{P}_1\)} & \multirow{2}{*}{\(1^{st}\) } & \(0\) & \(\{\ket{0}\bra{0},\ket{1}\bra{1}\}\) &  \\ \cline{3-4}
                  &                   & \(1\) & \(\left\{\ket{x^{01}_+}\bra{x^{01}_+},\ket{x^{01}_-}\bra{x^{01}_-}\right\}\) &  \\ \cline{1-4}
\multirow{2}{*}{\(\mathbb{P}_2\)} & \multirow{2}{*}{\(2^{nd}\) } & \(0\) & \(\{\ket{2}\bra{2},\ket{3}\bra{3}\}\) &  \\ \cline{3-4}
                  &                   & \(1\) & \(\{\ket{x^{23}_+}\bra{x^{23}_+},\ket{x^{23}_-}\bra{x^{23}_-}\}\) &  \\ \cline{1-4}
\multirow{2}{*}{\(\mathbb{P}_3\)} & \multirow{2}{*}{\(3^{rd}\) } & \(0\) & \(\{\ket{4}\bra{4},\ket{5}\bra{5}\}\) &  \\ \cline{3-4}
                  &                   & \(1\) & \(\{\ket{x^{45}_+}\bra{x^{45}_+},\ket{x^{45}_-}\bra{x^{45}_-}\}\) &  \\ \cline{1-4}
\end{tabular}
\caption{Bob selects the \(i^{th}\) communication line if rank-2 projector \(\mathbb{P}_i\) clicks in his first measurement. Then based on the communication received from Alice through the respective classical channel, he chooses his final measurement.}\label{tab1}
\end{center}
\end{table}
This protocol exactly reproduces the measurement statistics at Bob's end, while utilizing three classical bits from Alice.

\section{Proof of Theorem \ref{theo4}}\label{apptheo4}
\noindent We start by recalling a definition from \cite{Watrous2018} (see Section {\bf 2.3.3} in page {\bf 113}). 
\begin{definition}\label{def1}
[Rank-1 extremal POVM] A $k$ outcome POVM \(\mathrm{M}\equiv\{\Pi_a\}_{a=1}^k\) is called rank-1 extremal POVM if for all \(a,~\Pi_a=p_a\mathrm{P}_a\), with \(p_a\geq0~\&~\mathrm{P}_a\) being a rank-1 projector, and \(\sum_ar_a\mathrm{P}_a=0\) implies \(r_ap_a=0~\forall~a\); or equivalently all nonzero elements in \(\{\Pi_a\}_{a=1}^k\) are linearly independent of each other. Let, \(\mathcal{M}^{ext}_{R1}\) denotes the set of all rank-1 extremal POVMs.
\end{definition}
\noindent The notion of rank-1 extremal POVMs leads us to the following useful Lemma. 
\begin{Lemma}\label{lemma2}
Any finite element rank-1 POVM \(\mathrm{M}_{R1}\equiv\{s_a\mathrm{P}_a\}_{a=1}^k\) can be written as probabilistic mixture of finite number of rank-1 extremal POVMs, i.e., \(\forall~a,~s_a\mathrm{P}_a=\sum_{\lambda=1}^{L<\infty}\mu_{\lambda} s_a^{\lambda}\mathrm{P}_a\), with \(\sum_{\lambda=1}^{L}\mu_{\lambda}=1\) and \(\forall~\lambda,~\mathrm{M}^{\lambda}\equiv\{s_a^{\lambda}\mathrm{P}_a\}_{a=1}^k\in\mathcal{M}^{ext}_{R1}\).
\end{Lemma}
\begin{proof}
Consider an arbitrary rank-1 POVM with finite outcomes \(\mathrm{M}_{R1}\equiv\{s_a\mathrm{P}_a\}_{a=1}^k\), with \(s_a\geq0\). According to Definition \ref{def1}, \(\mathrm{M}_{R1}\) allows convex decomposition in terms of \(\mathrm{M}^{\lambda}\equiv\{{s_a^{\lambda}\mathrm{P}_a^{\lambda}}\}_{a=1}^{k}\in\mathcal{M}^{ext}_{RI}\), i.e.
\begin{align}
s_a\mathrm{P}_a=\int_{\lambda}d\lambda \mu_{\lambda}s_a^{\lambda}\mathrm{P}_a^{\lambda};~~\mu_{\lambda}>0~~\&~\int_{\lambda}d\lambda \mu_{\lambda}=1.\label{int}
\end{align}
\(\mathrm{P}_a\) being a rank-1 projector it follows that \(\mathrm{P}_a^{\lambda}=\mathrm{P}_a\), whenever \(s_a^{\lambda}>0\). On the other hand, for \(s_a^{\lambda}=0\) also we can assume \(\mathrm{P}_a^{\lambda}=\mathrm{P}_a\), which thus implies \(\mathrm{P}_a^{\lambda}=\mathrm{P}_a,~\forall~a,\lambda\). Thus we have \(\mathrm{M}^{\lambda}\equiv\{s_a^{\lambda}\mathrm{P}_a\}_{a=1}^k\). For such an \(\mathrm{M}^{\lambda}\) we can define \(\mathcal{A}_{\lambda}:=\{a~|~s_a^{\lambda}>0\}\). The extremality of $M^{\lambda}$ implies the set of effects \(\{s_a^{\lambda}\mathrm{P}_a~|~a\in \mathcal{A}_{\lambda}\}\) to be linearly independent, and furthermore the condition $\sum_{a\in \mathcal{A}_{\lambda}}s_a^{\lambda}\mathrm{P}_a=\mathbf{I}$ uniquely specifies the values of $s_a^{\lambda}$'s for any $\mathcal{A}_{\lambda}$. As the set $\{\mathrm{P}_a\}_{a=1}^k$ contains finitely many projectors, there are only finitely many ways of choosing $\mathcal{A}_{\lambda}$ such that \(M^{\lambda}\) turns out to be a rank-1 extremal POVM. Therefore, the integral in Eq.(\ref{int}) gets replaced by finite summation, meaning 
\begin{align}
s_a\mathrm{P}_a=\sum_{\lambda=1}^L\mu_{\lambda} s_a^{\lambda}\mathrm{P}_a,~\text{with}~\mu_{\lambda}>0~\&~\sum_{\lambda=1}^L\mu_{\lambda}=1.
\end{align}
This completes the proof.
\end{proof}
\noindent With this, we now proceed to prove Theorem \ref{theo4}.
\begin{proof}
Since any separable POVM is coarse-graining of rank-1 product POVMs, it suffices to prove our claim for the later only. Consider an $K$ outcomes rank-1 product POVM  
\begin{align}
 \mathrm{M}\equiv\left\{p_i\mathrm{P}_{u_i}\otimes\mathrm{P}_{v_i}~|~\ket{u_i}\in\mathbb{C}^{d_1}_A,~\ket{v_i}\in\mathbb{C}^{d_2}_B\right\}_{i=1}^K.
\end{align}
Given a known state \(\ket{\psi}\in\mathbb{C}^{d_1}\) to Alice and an unknown state \(\ket{\phi}\in\mathbb{C}^{d_2}\) to Bob, they aim to reproduce the outcome statistics 
\begin{align}
p(i|\psi,\phi):=p_i\Tr[\mathrm{P}_{u_i}\mathrm{P}_\psi]\Tr[\mathrm{P}_{v_i}\mathrm{P}_\phi]\label{sim}
\end{align}
at Bob's laboratory. Denoting \(p(i,\psi)=p_i\Tr[\mathrm{P}_{u_i}\mathrm{P}_\psi]\), the statistics in Eq.(\ref{sim}) can be viewed as the outcome statistics of the the effective rank-1 POVM \(\mathrm{M}^\psi_{R1}:=\{p(i,\psi)\mathrm{P}_{v_i}\}_{i=1}^K\) on Bob's unknown state $\ket{\phi}$. Lemma \ref{lemma2} ensures that POVM \(\mathrm{M}^\psi_{R1}:=\{p(i,\psi)\mathrm{P}_{v_i}\}_{i=1}^K\) can be expressed as probabilistic mixture of finite number of rank-1 extremal POVMs \(\mathrm{M}^\lambda\equiv\{s^\lambda_i\mathrm{P}_{v_i}\}_{i=1}^K\), i.e.
\begin{align}
p(i,\psi)\mathrm{P}_{v_i}=\sum_{\lambda=1}^L\mu_{\lambda}(\psi)s_i^{\lambda}\mathrm{P}_{v_i}.\label{mix}
\end{align}
More specifically Eq.(\ref{mix}) depicts that the coefficients \(\{\mu_{\lambda}(\psi)\}_{\lambda=1}^L\) in convex mixture depend on the state of Alice's system. To simulate the statistics of Eq.(\ref{sim}) at Bob's end, Alice  given a known state $\psi$ generates a random variable \(\lambda\in\{1,2,\cdots,L\}\) according to probability distribution $\{\mu_{\lambda}(\psi)\}_{\lambda=1}^L$ and communicates it to Bob using \(\log L\)-bits of classical communication. Upon receiving the random variable \(\lambda\) Bob implements the corresponding rank-1 extremal POVM \(\mathrm{M}^{\lambda}\) on his unknown state \(\phi\). This completes the protocol.  
\end{proof}

\subsection{Classical Simulation of twisted-Butterfly POVM}\label{appbutterfly}
The twisted-Butterfly POVM \(\mathrm{M}_{tb}\) induces the following effective POVM on Bob's part:
\begin{align}
\mathrm{M}_{tb}^e\equiv\left\{\!\begin{aligned}
\Pi^e_1&:=\frac{1}{2}
(1+\psi_z)\mathrm{P}_{\hat{z}^\perp},~~
\Pi^e_{21}:=\frac{3}{8}
(1-\frac{2\sqrt{2}}{3}\psi_x+\frac{1}{3}\psi_z)\mathrm{P}_{\hat{z}},\\
\Pi^e_{22}&:=\frac{3}{8},
(1-\psi_z)\mathrm{P}_{\hat{\alpha}},~~
\Pi^e_{31}:=\frac{3}{8}
(1-\frac{2\sqrt{2}}{3}\psi_x+\frac{1}{3}\psi_z)\mathrm{P}_{\hat{z}},\\
&\hspace{2cm}\Pi^e_{32}:=\frac{3}{8}
(1-\psi_z)\mathrm{P}_{\hat{\beta}^\perp}
\end{aligned}\right\},  
\end{align}
where \(\psi_z\) denotes the \(z\) component of the Bloch vector of of Alice's known state \(\ket{\psi}\). Using the four projectors \(\left\{\mathrm{P}_{\hat{z}},\mathrm{P}_{\hat{z}^\perp},\mathrm{P}_{\hat{\alpha}},\mathrm{P}_{\hat{\beta}^\perp}\right\}\) one can obtain only four rank-1 extremal POVMs, namely
\begin{align}
\left\{\begin{aligned}
\mathrm{M}^1&:=\left\{\mathrm{P}_{\hat{z}^\perp},\mathrm{P}_{\hat{z}},0,0,0\right\},~~\hspace{1.3cm}
\mathrm{M}^2:=\left\{\mathrm{P}_{\hat{z}^\perp},0,0,\mathrm{P}_{\hat{z}},0\right\},\\
\mathrm{M}^3&:=\left\{0,\frac{1}{2}\mathrm{P}_{\hat{z}},\frac{3}{4}\mathrm{P}_{\hat{\alpha}},0,\frac{3}{4}\mathrm{P}_{\hat{\beta}^\perp}\right\},~~
\mathrm{M}^4:=\left\{0,0,\frac{3}{4}\mathrm{P}_{\hat{\alpha}},\frac{1}{2}\mathrm{P}_{\hat{z}},\frac{3}{4}\mathrm{P}_{\hat{\beta}^\perp}\right\}
\end{aligned}\right\}, 
\end{align} 
The POVM \(\mathrm{M}_{tb}^e\) allows a convex decomposition in terms of extremal POVMs \(\{\mathrm{M}^\lambda\}_{\lambda=1}^4\), i.e.
\begin{align}
\mathrm{M}_{tb}^e&=\sum_{i=1}^4\mu_i(\psi) \mathrm{M}^i,
\end{align}
where,
\begin{align}
\left.\begin{aligned}
\mu_{1}(\psi)&=\max\{0,\frac{1}{8}(1-2\sqrt{2}\psi_x+3\psi_z)\},~~~~~
\mu_{2}(\psi)=\frac{1}{2}(1+\psi_z)-\mu_{1}(\psi),\\
\mu_{3}(\psi)&=\frac{3}{4}
(1-\frac{2\sqrt{2}}{3}\psi_x+\frac{1}{3}\psi_z)-2\mu_{1}(\psi),~~
\mu_{4}(\psi)=\frac{3}{4}
(1+\frac{2\sqrt{2}}{3}\psi_x+\frac{1}{3}\psi_z)-2\mu_{2}(\psi)
\end{aligned}\right\}.
\end{align}
It is easy to verify that  $\{\mu_{\lambda}(\psi)\}_{\lambda=1}^4$ is indeed a probability distribution $\forall~\psi$. To simulate statistics of twisted-butterfly POVM, Alice after receiving classical description of the state $\psi$ communicates a four valued random variable \(\{\lambda\}_{\lambda=1}^4\) sampled according to a distribution $\{\mu_{\lambda}(\psi)\}_{\lambda=1}^4$, and then Bob accordingly performs the measurement  \(\mathrm{M}^{\lambda}\) on his unknown state $\phi$. Thus 2 bits of communication channel is required from Alice to Bob to implement the classical simulation protocol. 

\begin{remark}\label{rmk1}
While the measurement \(\mathrm{M}^{(A)}_{twist}\) is LOCC-implementable, the measurement \(\mathrm{M}_{tb}\) is not implementable via LOCC (Lemma $1$). However, simulation of the outcome statistics on Bob's end for both the measurements is not possible with 1 bit classical communication from Alice, but possible with 2 bits of communication.
\end{remark}

\section{Simulation of multipartite fully separable measurements}\label{appmultisep}
\noindent We start by introducing a natural multipartite generalization of the channel simulation task, that invokes more than one senders (say) Alice-$1$, Alice-$2\cdots$ Alice-$n$ and one receiver Bob. Formally the task is defined as follows:
\begin{itemize}
\item Each of the senders is given classical description of qubit state, i.e., Alice-$i$ receives classical description of the state $\psi_{A_i}\in\mathbb{C}^{d_i}_{A_i}$. Importantly, the knowledge of the state is known only to the $i^{th}$ Alice, while it is oblivious to all other senders and the receiver.
\item Bob receives an unknown quantum state $\phi_B\in\mathbb{C}^{d}_B$. Likewise the one sender-one receiver case, here also the state $\phi_B$ is unknown to the senders.
\item Bob aims to reproduce statistics of a $K$ outcome multipartite measurement $\mathrm{M}_{A_1\cdots A_nB}\equiv\left\{\Pi^b_{A_1\cdots A_nB}\right\}_{b=1}^K$ on the state $\bigotimes_{i=1}^n\psi_{A_i}\otimes\phi_B$, which reads as
\begin{align}
p\left(b|\bigotimes_{i=1}^n\psi_{A_i}\otimes\phi_B\right)=\Tr\left[\left(\bigotimes_{i=1}^n\psi_{A_i}\otimes\phi_B\right)\Pi^b_{A_1\cdots A_nB}\right].\label{multi}
\end{align}
\end{itemize}
Naturally, this raises the question of whether \textbf{Theorem 4} can be generalized to fully separable measurements composed exclusively of fully separable effects \cite{Horodecki2009}. Before proceeding further, it is important to note that in the multipartite setting, such measurements can exhibit nonlocal behavior in the sense that they may not be implementable within the operational paradigm of LOCC. A canonical example is the three-qubit Shift basis measurement \cite{Bennett1999(1)}:
\begin{align}
\mathrm{M}_{Shift}\equiv \left\{\!\begin{aligned} \ket{000}\bra{000},~~~~~\ket{111}\bra{111},~~~~~\ket{+01}\bra{+01},~\ket{-01}\bra{-01},\\
\ket{01+}\bra{01+},~\ket{01-}\bra{01-},~\ket{1+0}\bra{1+0},~\ket{1-0}\bra{1-0}
\end{aligned}\right\},
\end{align} 
where $\ket{\pm}:=\frac{1}{\sqrt{2}}(\ket{0}\pm\ket{1}$. Notably, Bennett et al. \cite{Bennett1999(1)} demonstrated that the shift-basis measurement cannot be implemented via LOCC when all parties are spatially separated, though it becomes feasible if any two parties are in one laboratory. This construction was later extended to multipartite and higher-dimensional systems \cite{Niset2006}. More recently, fully separable measurements have been identified that remain non-implementable under LOCC unless all parties are co-located, revealing a stronger form of measurement nonlocality \cite{Halder2019,Rout2019,Rout2021}. This motivates an investigation into the applicability of Theorem 4 to such fully separable measurements. \\\\
\noindent While dealing with classical simulation of the statistics in Eq.(\ref{multi}), the following two configurations arise depending on how the classical resources are allowed:
\begin{itemize}
\item \textbf{Configuration A:} All the parties (the senders and the receiver) can exchange arbitrarily large but finite amount of classical communication among one another.
\item \textbf{Configuration B:} Arbitrarily large but finite amount one-way classical communicate is allowed from each of the senders to the receiver. Back ward communication from the receiver to the senders as well as communication among the senders are not allowed. 
\end{itemize}
Within the {\bf Configuration A}, in the following we first show that it is possible to generalize the {\bf Theorem 4}.
\begin{manualtheorem}{4S-A}\label{theo4SA}
Statistics of any fully separable measurement on a quantum states $\psi_{A_i}$, which is known to $i^{th}$ Alice but unknown to others, and an unknown state provided to Bob, can always be simulated at Bob's end by finite classical communication allowed within Configuration A. 
\end{manualtheorem}
\begin{proof}
The proof proceeds by induction. Assuming the result holds for $n-1$ senders, we show it must also hold for $n$. As in the bipartite case, it suffices to consider separable measurements performed by Bob, specifically those composed of rank-1 POVM elements. Let us consider a $K$-outcome, rank-1, fully product POVM of the form
\begin{align}
\mathrm{M}\equiv\left\{p_b\bigotimes_{i=1}^n\mathrm{P}_{u^i_b}\otimes\mathrm{P}_{v_b}~|~\ket{u^i_b}\in\mathbb{C}^{d_i}_{A_i},~\ket{v_b}\in\mathbb{C}^{d}_B\right\}_{b=1}^K.
\end{align}
The outcome statistics of Eq.(\ref{multi}) reads as
\begin{align}
p\left(b|\bigotimes_{i=1}^n\psi_{A_i}\otimes\phi_B\right)&=\Tr\left[\left(\bigotimes_{i=1}^n\psi_{A_i}\otimes\phi_B\right)\left(p_b\bigotimes_{i=1}^n\mathrm{P}_{u^i_b}\otimes \mathrm{P}_{v_b}\right)\right]\nonumber\\
&=p_b\Tr\left[\psi_{A_1}\mathrm{P}_{u^1_b}\right]\times\Tr\left[\left(\bigotimes_{i=2}^n\psi_{A_i}\otimes\phi_B\right)\left(\bigotimes_{i=2}^n\mathrm{P}_{u^i_b}\otimes \mathrm{P}_{v_b}\right)\right].
\end{align}
Denoting \(p(b,\psi_{A_1})=p_b\Tr[\psi_{A_1}\mathrm{P}_{u^1_b}]\), the statistics in the above equation can be viewed as an effective rank-1 POVM separable measurement \(\mathrm{M}^{\psi_{A_1}}_{R1}:=\{p(b,\psi_{A_1})\bigotimes_{i=2}^n\mathrm{P}_{u^i_b}\otimes \mathrm{P}_{v_b}\}_{b=1}^K\) on the state $\bigotimes_{i=2}^n\psi_{A_i}\otimes\phi_B$. {\bf Lemma 2} (in main manuscript) ensures that POVM \(\mathrm{M}^{\psi_{A_1}}_{R1}\) can be expressed as probabilistic mixture of finite number of rank-1 extremal POVMs \(\mathrm{M}^\lambda\equiv\{s^\lambda_b\bigotimes_{i=2}^n\mathrm{P}_{u^i_b}\otimes \mathrm{P}_{v_b}\}_{b=1}^K\), i.e.
\begin{align}
p(b,\psi_{A_1})\bigotimes_{i=2}^n\mathrm{P}_{u^i_b}\otimes \mathrm{P}_{v_b}=\sum_{\lambda=1}^L\mu_{\lambda}(\psi_{A_1})s_b^{\lambda}\bigotimes_{i=2}^n\mathrm{P}_{u^i_b}\otimes \mathrm{P}_{v_b}.
\end{align}
To classically simulate the required statistics, Alice-$1$ tosses a coin \(\{\mu_{\lambda}(\psi_{A_1})\}_{\lambda=1}^L\) and communicates the outcome $\lambda$ to the rest of the parties indicating that the rest of the parties should implement the measurement \(\mathrm{M}^\lambda\equiv\{s^\lambda_b\bigotimes_{i=2}^n\mathrm{P}_{u^i_b}\otimes \mathrm{P}_{v_b}\}_{b=1}^K\) on their joint state $\bigotimes_{i=2}^n\psi_{A_i}\otimes\phi_B$. Thus for each $\lambda$ we now have a similar separable measurement problem among the $n-1$ senders and one receiver. According to our inductive hypothesis every such measurement $\mathrm{M}^\lambda$ can be implemented by finite communication among the $n-1$ Alice's and Bob. Thus the measurement  $\mathrm{M}^{\psi_{A_1}}_{R1}$ can also be simulated among the $n$ senders and one receiver. The base case of the inductive proof for one sender and and one receiver follows from the \textbf{Theorem 4}. This completes the proof.
\end{proof}
At this point, one might suspect that under restricted communication scenario (i.e., \textbf{Configuration B}) Theorem \ref{theo4SA} may no longer hold. As we argue now this is not the case. 
\begin{manualtheorem}{4S-B}\label{theo4SB}
Statistics of any fully separable measurement on a quantum states $\psi_{A_i}$, which is known to $i^{th}$ Alice but unknown to others, and an unknown state provided to Bob, can always be simulated at Bob's end by finite one-way classical communication from the senders to the receiver as allowed within Configuration B. 
\end{manualtheorem}
\begin{proof}
We detail the argument for tripartite case, and the generalization follows for higher number of senders. Given two senders Alice-$1$ and Alice-$2$ and a receiver Bob let us consider a $K$-outcome, rank-1, fully product POVM of the form
\begin{align}
\mathrm{M}\equiv\left\{p_b\mathrm{P}_{u^1_b}\otimes\mathrm{P}_{u^2_b}\otimes\mathrm{P}_{v_b}~|~\ket{u^1_b}\in\mathbb{C}^{d_1}_{A_1},\ket{u^2_b}\in\mathbb{C}^{d_2}_{A_2},~\ket{v_b}\in\mathbb{C}^{d}_B\right\}_{b=1}^K.
\end{align}
\begin{itemize}
\item[] As before, the effective POVM \(\mathrm{M}^{\psi_{A_1}}_{R1}:=\{p(b,\psi_{A_1})\mathrm{P}_{u^2_b}\otimes \mathrm{P}_{v_b}\}_{b=1}^K\) on the state $\psi_{A_2}\otimes\phi_B$ can be written as a probabilistic mixture of finite number of rank-1 extremal POVMs \(\mathrm{M}^\lambda\equiv\{s^\lambda_b\mathrm{P}_{u^2_b}\otimes \mathrm{P}_{v_b}\}_{b=1}^K\).
\item[] In the case of {\bf Configuration A}, Alice-$1$ communicates the information of $\lambda$ to Alice-$2$ as well as Bob indicating them to simulate the statistics of \(\mathrm{M}^\lambda\). Whenever no communication between the senders is allowed, could just send the information corresponding to every possible value of $\lambda\in\{1\cdots L\}$ since she has the complete classical description of the state $\psi_{A_2}$. Therefore, Alice-$2$ does not require the knowledge of the measurement $\mathrm{M}^\lambda$ to be simulated.
\item[ ] Therefore, the protocol in Theorem \ref{theo4SA} can be modified by limiting Alice-$1$'s communication of $\lambda$ only to Bob. 
\item[] Since Bob receives information for every possible value of $\lambda$ from Alice-$2$, he can suitably choose the relevant information as indicated to him by Alice-$1$. 
\end{itemize}
Following the aforementioned protocol, the tripartite statistics can be achieved in the case where the senders do not communicate between them. However, this comes at the cost of a large (but finite) amount communication from Alice-$2$ to Bob. It is not hard to see that the above argument generalizes for more than two senders. In this case, $i^{th}$ Alice communicates all relevant information to Bob corresponding to Alice-$1,\cdots, $ Alice-$(i-1)$.
\end{proof}

\section{Detailed proof of Theorem \ref{theo5}}\label{apptheo5}
Given the state \(\psi := \frac{1}{2}(\mathbf{I}_2 + \hat{\psi} \cdot \sigma)\), the protocol ensures that the state \(\mathcal{X} \mathrm{P}_{\hat{\omega}_{i^\star}} \mathcal{X}^\dagger\), prepared at Bob's end, lies within the cone forming an apex angle \(\theta_m\) with the vector \(\hat{\psi}\) (see Fig.\ref{fig1s}). Furthermore, since the random variable \(\mathcal{X}\) is drawn Haar-randomly from the set of unitaries acting on \(\mathbb{C}^2\), all the states within this cone are prepared with equal probability. Consequently, on average, Bob prepares a resulting density operator \(\rho_{R}(\psi) = \frac{1}{2}(\mathbf{I}_2 + \vec{\psi}^R \cdot \sigma)\), with \(\vec{\psi}^R = (\psi^R_x, \psi^R_y, \psi^R_z)^{\mathrm{T}} \in \mathbb{R}^3\).
\begin{figure}[b!]
\centering
\includegraphics[scale=0.42]{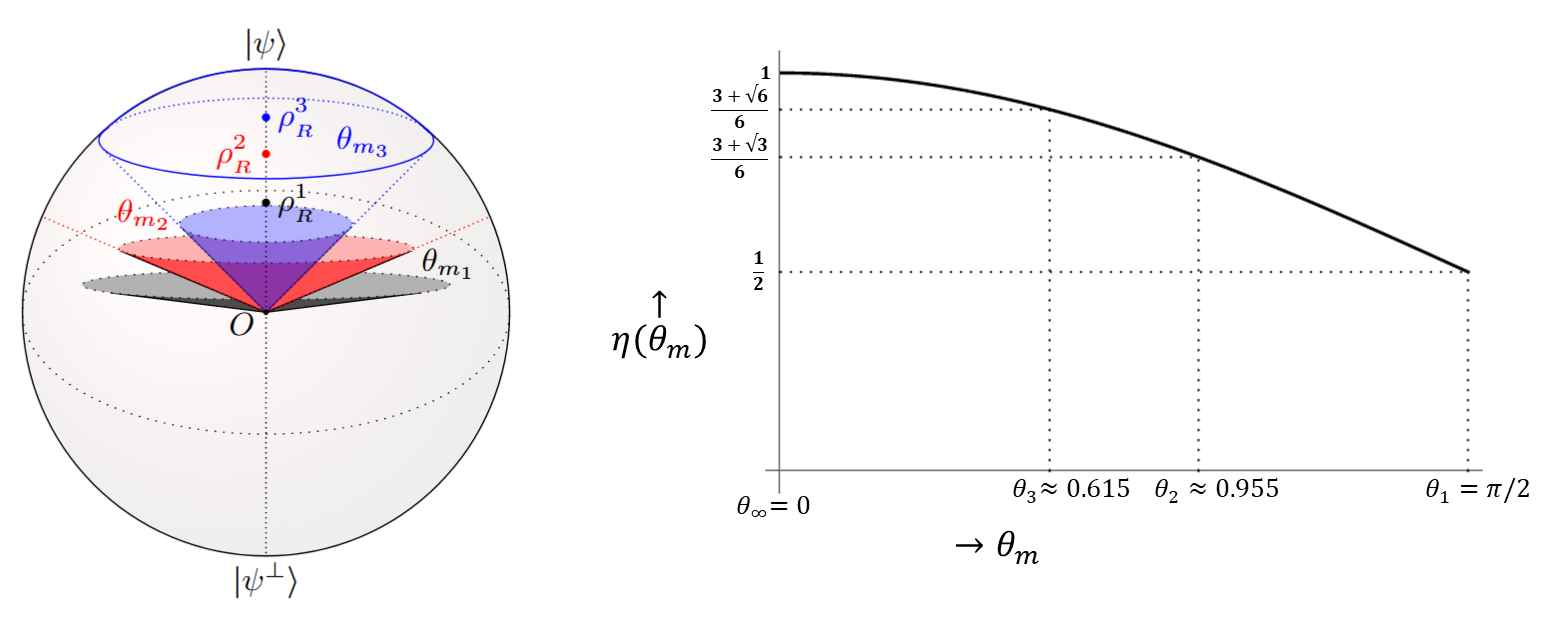}
\caption{(Color online) [Left] Generic Simulation of the Qubit Depolarizing Channel $\mathrm{D}_\eta : \mathcal{D}(\mathbb{C}^2) \to \mathcal{D}(\mathbb{C}^2)$. Given a qubit state $\psi = \frac{1}{2}(\mathbf{I}_2 + \hat{\psi} \cdot \sigma)$ to Alice, the state prepared at Bob’s end is of the form $\mathcal{X} \mathrm{P}_{\hat{\omega}_{i^\star}} \mathcal{X}^\dagger$, where $\mathcal{X}$ is a Haar-random unitary on $\mathbb{C}^2$, and the vector $\hat{\omega}_{i^\star}$ lies within a cone of half apex angle $\theta_m$ centered around $\hat{\psi}$. Averaging over the random unitaries $\mathcal{X}$, the state at Bob's end becomes $\frac{1}{2}(\mathbf{I}_2 + \eta(\theta_m)\, \hat{\psi} \cdot \sigma)$, effectively simulating a depolarizing channel with parameter $\eta(\theta_m)$. As the communication $m$ increases, the apex angle $\theta_m$ decreases, leading to a higher value of $\eta(\theta_m)$, and hence a less noisy depolarizing channel. Here we illustrate three such cases (not to scale), shown respectively in black, red, and blue, corresponding to increasing communication levels $m_1 < m_2 < m_3$. [Right] Solid curve depicts variation of $\eta(\theta_m)$ with $\theta_m$ [see Eq.(\ref{plot})]. Values of $(\theta_m,\eta(\theta_m))$ for 1-bit, 2-bit, and 3-bit communication are shown.}\vspace{-.2cm}
\label{fig1s}
\end{figure}

Denoting the Bloch vector of \(\mathcal{X} \mathrm{P}_{\hat{\omega}_{i^\star}} \mathcal{X}^\dagger\) as \((\sin\theta \cos\varphi, \sin\theta \sin\varphi, \cos\theta)^{\mathrm{T}} \in \mathbb{R}^3\), the components of the resulting density operator are given by:
\begin{align}
\left.\begin{aligned}
\psi^R_x &= \frac{\int_{0}^{\theta_m} \sin\theta \, d\theta \int_{0}^{2\pi} d\varphi \, \sin\theta \cos\varphi}{\int_{0}^{\theta_m} \sin\theta \, d\theta \int_{0}^{2\pi} d\varphi}, \\
\psi^R_y &= \frac{\int_{0}^{\theta_m} \sin\theta \, d\theta \int_{0}^{2\pi} d\varphi \, \sin\theta \sin\varphi}{\int_{0}^{\theta_m} \sin\theta \, d\theta \int_{0}^{2\pi} d\varphi}, \\
\psi^R_z &= \frac{\int_{0}^{\theta_m} \sin\theta \, d\theta \int_{0}^{2\pi} d\varphi \, \cos\theta}{\int_{0}^{\theta_m} \sin\theta \, d\theta \int_{0}^{2\pi} d\varphi}
\end{aligned}\right\}.
\end{align}
For instance, if Alice is given the state \(\ket{0}\bra{0} = \frac{1}{2}(\mathbf{I}_2 + \sigma_z)\), then the components of the Bloch vector for Bob's resulting state are:
\begin{align}
\psi^R_x = \psi^R_y = 0,~ \&~ \psi^R_z = \frac{1}{4} \cdot \frac{1 - \cos 2\theta_m}{1 - \cos \theta_m} := \eta(\theta_m). \label{plot} 
\end{align}
Thus, the resulting density operator at Bob's end is:
\begin{align}
\rho_R(\ket{0}) = \frac{1}{2}\left(\mathbf{I}_2 + \eta(\theta_m) \, \sigma_z\right).   
\end{align}
The above calculation yields the same result for any arbitrary state \(\psi\) provided to Alice, confirming that the protocol simulates a depolarizing channel with parameter \(\eta(\theta_m)\). Notably, as \(m\) increases, \(\theta_m\) decreases, and \(\eta(\theta_m)\) increases. In the limiting case \(m \to \infty\), we have \(\theta_m \to 0\) and \(\eta(\theta_m) \to 1\), which aligns with the claims of Theorem $1$. Thus, for any value of \(\eta < 1\), a simulation is always achievable with \(m\) bits of finite communication, provided \(m\) is sufficiently large.

\twocolumngrid

%


\begin{thebibliography}{82}
\bibitem [{Sel()}]{Self1}%
  \BibitemOpen
  \bibfield  {title} {\bibinfo {title} {A quantum channel from {A}lice to {B}ob is a linear map $\mathbf{\Lambda}:\mathcal{L}(\mathcal{H}_{A})\to\mathcal{L}(\mathcal{H}_{B})$; where $\mathcal{L}(\star)$ denotes the space of linear operators on the corresponding {H}ilbert space. {P}ositivity of the map demands $\mathbf{\Lambda}(\chi_{A})\ge0,~\forall~0\le\chi_{A}\in\mathcal{L}(\mathcal{H}_{A})$; and complete positivity demands positivity of $\mathbf{\Lambda}\otimes \mathrm{id}_n$ forall $n$, where $\mathrm{id}_n$ denotes the identity map on the operator space of $n$ dimensional {H}ilbert space},\ }\href@noop {} {\ }\BibitemShut {NoStop}%
\end{thebibliography}

\end{document}